\documentclass[12pt]{article}

\usepackage[a4paper]{geometry}
\usepackage{amssymb}
\usepackage{amsmath}
\usepackage{amsthm}

\theoremstyle{plain}
\newtheorem{theorem}{Theorem}[section]
\newtheorem{proposition}[theorem]{Proposition}
\newtheorem{lemma}[theorem]{Lemma}

\theoremstyle{definition}

\theoremstyle{remark}

\numberwithin{equation}{section}

\title{Branching form of the resolvent at thresholds for multi-dimensional discrete Laplacians}
\author{Kenichi {\scshape Ito}\footnote{Graduate School of Mathematical Sciences, 
The University of Tokyo, 3-8-1 Komaba, Meguro-ku, Tokyo 153-8914, Japan.
E-mail: \texttt{ito@ms.u-tokyo.ac.jp}. 
}
\ \& 
Arne {\scshape Jensen}\footnote{Department of Mathematical Sciences,
Aalborg University, Skjernvej 4A, DK-9220 Aalborg \O{}, Denmark.
E-mail: \texttt{matarne@math.aau.dk}. 
}}
\date{}

\begin{document}
\allowdisplaybreaks
\maketitle

\begin{abstract}
We consider the discrete Laplacian on $\mathbb Z^d$,
and compute asymptotic expansions of its resolvent 
around thresholds embedded in continuous spectrum as well as those at end points.
We prove that
the resolvent has a square-root branching if $d$ is odd,
and a logarithm branching if $d$ is even,
and, moreover, obtain explicit expressions for these branching parts 
involving the Lauricella hypergeometric function. 
In order to analyze a non-degenerate threshold of general form 
we use an elementary step-by-step expansion procedure,
less dependent on special functions.
\end{abstract}


\medskip


\section{Introduction}
Consider the discrete Laplacian on $\mathbb Z^d$, given for any function $u\colon\mathbb Z^d\to\mathbb C$ by
\begin{equation}\label{standard}
(\triangle u)[n]=\sum_{j=1}^d(u[n+e_j]+u[n-e_j]-2u[n]),
\quad n\in\mathbb Z^d,
\end{equation}
where $\{e_j\}_{j=1,\dots,d}\subset\mathbb Z^d$ is the standard basis of $\mathbb Z^d$. The operator $H_0=-\triangle$ is bounded and self-adjoint on $\ell^2(\mathbb Z^d)$. In Fourier space $L^2(\mathbb T^d)$ it is given by multiplication by
\begin{equation*}
\Theta(\theta)=4\sin^2({\theta_1}/2)+\dots+4\sin^2({\theta_d}/2)
;\quad 
\theta=(\theta_1,\dots,\theta_d)\in\mathbb T^d.
\end{equation*}
The critical values of $\Theta(\theta)$ are $0$, $4$, \ldots, $4d$,
and they are called the \emph{thresholds} of $H_0$. The purpose of this paper is to investigate the asymptotic behavior of the resolvent
\begin{align*}
R_0(z)=(H_0-z)^{-1},
\end{align*}
or its convolution kernel
\begin{align}
k(z,n)=
(2\pi)^{-d}\int_{\mathbb T^d}\frac{\mathrm e^{\mathrm in\theta}}{\Theta(\theta)-z}\,\mathrm d\theta,
\label{170206}
\end{align}
as the spectral parameter $z\in\mathbb C\setminus \sigma(H_0)$ approaches one of the thresholds. 
We prove that there appears a square-root branching if $d$ is odd,
and a logarithm branching if $d$ is even.
Moreover we obtain 
explicit expressions for these branching parts in terms of the Lauricella hypergeometric function of type $B$. 
As far as the authors are aware, an explicit expression of the integral \eqref{170206} 
by special functions seems not to have been known,
and this paper provides a partial answer to this question.

By localizing in Fourier space in a neighborhood of a threshold and 
changing variables appropriately one finds that the problem
gets very similar to an analysis of
the resolvent of an ultra-hyperbolic operator of signature $(p,q)$ with $p+q=d$: 
\begin{equation*}
\square_{p,q}=\partial^2_1+\cdots+\partial^2_p-\partial^2_{p+1}-\cdots
-\partial^2_{p+q}. 
\end{equation*}
Thus a large part of this paper is concerned with a study of the resolvent of $\square_{p,q}$ 
at its threshold zero, for all possible values of $p$ and $q$ with $p+q=d$. 
The kernel of the resolvent of $\square_{p,q}$ can be written explicitly 
using a Macdonald function, see \cite{Brychkov-Prudnikov}.
However, with an application to the discrete Laplacian in mind,
we need to develop another more elementary and general expansion scheme 
less dependent on a particular special function.

In the following section we describe the setting in more detail and then state the results.
Here we comment on the literature.

Whereas there is an extensive literature on ultra-hyperbolic operators, see e.g.\ \cite{Hormander},
the question concerning the behavior of the resolvent 
around its hyperbolic thresholds raised here seems not to have been studied previously except 
in a few cases:
Murata \cite{Murata} computed an expansion for 
a differential operator of constant coefficients on $\mathbb R^d$,
and Komech, Kopylova, and Vainberg \cite{Komech-Kopylova-Vainberg}
computed an expansion for the discrete Laplacian on $\mathbb Z^2$.
However, detailed computations seem not to be given there,
and all the coefficients are not determined completely, either.
Our paper provides the details of computations quite clearly, 
and for the first time determines all the coefficients of branching part of the resolvent for 
the discrete Laplacian.

An expansion of the resolvent for the discrete Laplacian
could be a very important contribution 
to the direct and inverse scattering problems in the discrete spaces. 
In fact, in the discrete setting, typical techniques in the continuous spaces,
such as the limiting absorption principle or some reconstruction schemes, 
often fail due to embedded thresholds, cf.\ \cite{Poulin} and \cite{Isozaki-Korotyaev}. 
Hence the embedded thresholds of the discrete operators are recently of increasing interest.

One proposal to deal with the embedded thresholds of the discrete Laplacian~\eqref{standard} has been to replace it with a different version, see~\cite{Poulin} and references therein.

We also mention the recent result~\cite{Richard-Tiedra}, where an asymptotic expansion of the resolvent 
around embedded thresholds for an operator on a waveguide is obtained.
Note that the embedded thresholds in a waveguide have a different nature from those considered here. 
They originate from elliptic critical points of multiple scattering channels 
similar to $N$-body quantum systems rather than those from hyperbolic critical points.

\section{Setting and result}

\subsection{Discrete Laplacian}\label{1606020}

For any function $u\colon \mathbb Z^d \to\mathbb C$ we define 
$\triangle u\colon\mathbb Z^d\to\mathbb C$ by 
\begin{align*}
(\triangle u)[n]=\sum_{j=1}^d(u[n+e_j]+u[n-e_j]-2u[n]),\quad n\in\mathbb Z^d,
\end{align*}
where $\{e_j\}_{j=1,\dots,d}\subset\mathbb Z^d$ is the standard basis. 
The operator $H_0=-\triangle$ is bounded and self-adjoint on 
the Hilbert space $\mathcal H=\ell^2(\mathbb Z^d)$, and 
its spectrum is given by 
\begin{align*}
\sigma(H_0)=\sigma_{\mathrm{ac}}(H_0)=[0,4d].
\end{align*}
In fact, letting $\mathbb T=\mathbb R/(2\pi\mathbb Z)$
and $\widehat{\mathcal H}= L^2(\mathbb {T}^d)$,
we define the Fourier transform 
${\mathcal F}\colon \mathcal H \to\widehat{\mathcal H}$ 
and its inverse 
${\mathcal F}^*\colon \widehat{\mathcal H}\to\mathcal H$  by 
\begin{align*}
({\mathcal F} u)(\theta)=(2\pi)^{-d/2}\sum_{n\in\mathbb {Z}^d}\mathrm e^{-\mathrm in\theta}u[n],
\quad
({\mathcal F}^* f)[n]=(2\pi)^{-d/2}\int_{\mathbb T^d}\mathrm e^{\mathrm in\theta}f(\theta)\mathrm d\theta,
\end{align*}
and then $H_0$ is diagonalized as 
\begin{align*}
\mathcal FH_0\mathcal F^*=\Theta(\theta);\quad
\Theta(\theta)=4\sin^2({\theta_1}/2)+\dots+4\sin^2({\theta_d}/2),
\end{align*}
where the right-hand side $\Theta(\theta)$ denotes a multiplication
operator by the same function. 
The critical values $0,4,\dots,4d$ of $\Theta(\theta)$ 
are called \textit{thresholds}, 
and it is \textit{of elliptic} or \textit{hyperbolic type} 
depending on the associated critical point being of elliptic or hyperbolic type, respectively.
We note that 
the critical points associated  
with a critical value $4q$, $q\in\{0,\dots,d\}$, form a set
$$\Omega(p,q)=\bigl\{\theta\in\{0,\pi\}^d;\ 
\#\{j;\ \theta_j=0\}=p,\ \#\{j;\ \theta_j=\pi\}=q\bigr\},$$
where $p=d-q$, and $\#$ denotes a number of elements of the set that follows.
Hence the thresholds $0,4d$ at end points are of elliptic type, and 
the embedded thresholds $4,\dots, 4(d-1)$ are of hyperbolic type.

Let us fix one threshold $4q$, $q\in\{0,\dots,d\}$, and 
investigate the asymptotic behavior of the resolvent $R_0(z)=(H_0-z)^{-1}$ for $z\sim 4q$.
Note that if we set 
for $(z,n)\in(\mathbb C\setminus [0,4d])\times\mathbb Z^d$
\begin{align}
k(z,n)=
(2\pi)^{-d}\int_{\mathbb T^d}\frac{\mathrm e^{\mathrm in\theta}}{\Theta(\theta)-z}\,\mathrm d\theta
,
\label{11.4.17.3.5b}
\end{align}
then for any rapidly decaying function $u\colon\mathbb Z^d\to\mathbb C$
\begin{align*}
R_0(z)u=k(z,{}\cdot{})* u.
\end{align*}
This reduces the problem to an investigation of the convolution kernel $k(z,n)$.
The branching of $k(z,n)$ comes only from 
integrations \eqref{11.4.17.3.5b} 
on neighborhoods of the critical points $\Omega(p,q)$. 
Let us index points of $\Omega(p,q)$ as 
\begin{align*}
\Omega(p,q)=\bigl\{\theta^{(l)}\bigr\}_{l=1,\dots,L};\quad 
L=\#\Omega(p,q)=\left(\genfrac{}{}{0pt}{}{d}{p}\right)
=\left(\genfrac{}{}{0pt}{}{d}{q}\right),
\end{align*}
which are ordered, e.g., lexicographically with 
\begin{align}
\theta^{(1)}=
(0,\dots,0,\pi,\dots,\pi)\in\mathbb T^p\oplus\mathbb T^q=\mathbb T^d.
\label{17021814}
\end{align}
We fix a neighborhood $U_1\subset \mathbb T^d$ of $\theta^{(1)}$ as follows:
Define the local coordinates 
$\xi(\theta)=(\xi'(\theta),\xi''(\theta))\in\mathbb R^p\oplus\mathbb R^q$ 
around $\theta^{(1)}\in\mathbb T^d$ 
by 
\begin{align}
\begin{split}
\xi'(\theta)&
=
(2\sin(\theta_1/2),\dots,2\sin(\theta_p/2)),
\\
\xi''(\theta)&
=
(2\cos(\theta_{p+1}/2),\dots,2\cos(\theta_{p+q}/2)),
\end{split}
\label{11.9.4.7.4}
\end{align}
and then set 
\begin{align*}
U_1&=\bigl\{\theta\in\mathbb T^d;\ |\xi'(\theta)|+|\xi''(\theta)|<2\bigr\}.
\end{align*}
Similarly, we can fix neighborhoods $U_l\subset\mathbb T^d$ of $\theta^{(l)}$, 
$l=2,\dots, L$, and then decompose
\begin{align}
k(z,n)
&=k_0(z,n)+k_1(z,n)+\dots+k_L(z,n)
\label{16060111}
\end{align}
by
\begin{align*}
k_0(z,n)&=k(z,n)-k_1(z,n)-\dots-k_L(z,n),\\
k_l(z,n)&=(2\pi)^{-d}
\int_{U_{l}}\frac{\mathrm e^{\mathrm in\theta}}{\Theta(\theta)-z}\,\mathrm d\theta
\quad\text{for }l=1,\dots,L.
\end{align*}
Since the integration region of $k_0(z,n)$ is uniformly away from $\Omega(p,q)$,
the function $k_0(w+4q,n)$ is analytic in $w\in\Delta(4):=\{w\in \mathbb C;\ |w|<4\}$.
Thus it suffices to investigate $k_l(z,n)$, $l=1,\dots,L$.

To state our main theorem let us introduce functions $E^{(l)}(w,n)$, $l=1,\dots,L$, 
associated with $k_l(w+4q,n)$.
For $l=1$, 
letting $F_B^{(d)}$ be the \textit{Lauricella hypergeometric function of type $B$}, 
we define for $(w,n)\in\Delta(4)\times \mathbb Z^d$
\begin{align*}
E^{(1)}(w,n)
&=\frac{(-1)^{|n''|}\mathrm i^q}{2^d\pi^{d/2}\Gamma(d/2)}F_B^{(d)}\bigl(
1/2-n,
1/2+n;{}
{d}/{2};w^{(1)}/4\bigr).
\end{align*}
Here we split $n=(n',n'')\in\mathbb Z^p\oplus\mathbb Z^q$, 
and denoted 
$1/2\pm n=(1/2\pm n_1,\dots,1/2\pm n_d)$ for simplicity and
\begin{align}
w^{(1)}=(w,\dots,w,-w,\dots,-w)\in \mathbb C^p\oplus\mathbb C^q.
\label{17021815}
\end{align}
Inserting the definition of $F_B^{(d)}$, see, e.g., \cite[\S 8.6]{Slater},
and noting the formulas holding for $m\in\mathbb Z_+:=\{0,1,2,\ldots\}$:
\begin{align}
\Gamma(1/2+m)=\frac{(2m-1)!!}{2^m}\sqrt{\pi},\quad
\Gamma(1/2-m)=\frac{(-2)^m}{(2m-1)!!}\sqrt{\pi},
\label{170218}
\end{align}
we can also write
\begin{align*}
E^{(1)}(w,n)
=\frac{(-1)^{|n'|}\mathrm i^q}{2^d\pi^{d/2+1}}
\sum_{\alpha\in\mathbb Z_+^d}
&\frac{\prod_{j=1}^d\bigl[\Gamma(1/2-n_j+\alpha_j)\Gamma(1/2+n_j+\alpha_j)\bigr]}{\Gamma(|\alpha|+d/2)}
\\&
\cdot\frac{(-1)^{|\alpha''|}w^{|\alpha|}}{4^{|\alpha|}\alpha!},
\end{align*}
where $\alpha=(\alpha',\alpha'')\in\mathbb Z_+^p\oplus\mathbb Z_+^q$. 
We note that the choice of signs in \eqref{17021815}
corresponds to the form of critical point \eqref{17021814}.
By changing signs of \eqref{17021815} appropriately 
we define $E^{(l)}(w,n)$, $l=2,\dots,d$, similarly to the above.
We note that $E^{(l)}(w,n)$ satisfies the eigenequation in $n$:
$$(-\triangle-w-4q)E^{(l)}(w,n)=0.$$
The last identity will be verified in Propositions~\ref{11.9.10.10.7} and \ref{16052215b}.

In this paper, unless otherwise noted, we usually 
choose the branch of $\sqrt w$ for $w\in \mathbb C\setminus [0,\infty)$ 
with $\mathop{\mathrm{Im}}\sqrt w>0$, so that 
for any $w\in\mathbb C_+=\{w\in\mathbb C;\ \mathop{\mathrm{Im}}w>0\}$ we have
\begin{align*}
\sqrt{-w}&=\mathrm i\sqrt w .
\end{align*}
On the other hand, a branch of $\log w$ for $w\in\mathbb C\setminus (-\infty,0]$
is always chosen such that $-\pi<\mathop{\mathrm{Im}}\log w<\pi$, so that 
for any $w\in\mathbb C\setminus (-\infty,0]$
\begin{align*}
\log (1/w)&=-\log w.
\end{align*}

Now we state the main theorem of this paper.

\begin{theorem}\label{1606011129}
Let $q\in\{0,1,\dots,d\}$.
The function $k(z,n)$ defined
by \eqref{11.4.17.3.5b} for $(z,n)\in(\mathbb C\setminus [0,4d])\times\mathbb Z^d$
has the following expressions:
\begin{enumerate}
\item
If $d$ is odd, then 
there exists a function $\chi(w,n)$ analytic in $w\in\Delta(4)$ such that 
\begin{align}
k(w+4q,n)=\mathrm i\pi(\sqrt w)^{d-2} \sum_{l=1}^LE^{(l)}(w,n)
+\chi(w,n);
\label{17021716}
\end{align}
\item
If $d$ is even, then
there exists a function $\chi(w,n)$ analytic in $w\in\Delta(4)$ such that 
\begin{align}
k(w+4q,n)
&=
-(\sqrt w)^{d-2}(\log w)
\sum_{l=1}^LE^{(l)}(w,n)
+\chi(w,n).
\label{17021717}
\end{align}

\end{enumerate}

\end{theorem}

Since $E^{(l)}(w,n)$ is analytic in $w\in\Delta(4)$,
the first terms on the right-hand sides of \eqref{17021716} and \eqref{17021717} 
are exactly the branching parts of the kernel $k(z,n)$.
Thus the resolvent $R_0(z)$ also has a square-root branching if $d$ is odd,
and a logarithm branching if $d$ is even.

Changing variables by (\ref{11.9.4.7.4}) and letting $w=z-4q$, 
we can write
\begin{align}
k_1(w+4q,n)
&=
(2\pi)^{-d}\int_{|\xi'|+|\xi''|<2}\frac{\mathrm e^{\mathrm in\theta(\xi)}
}{\xi'^2-\xi''^2-w}\,
\frac{\mathrm d\xi}{\prod_{j=1}^d(1-\xi_j^2/4)^{1/2}}
,
\label{160531}
\end{align}
where the branch of $(1-\xi_j^2/4)^{1/2}$ is chosen to be the principal one
with cut along the negative real axis.
Hence the proof of Theorem~\ref{1606011129} reduces to an expansion of \eqref{160531}.
The situation is very similar to an expansion of the resolvent for an ultra-hyperbolic operator on 
$\mathbb R^d$,
and in fact in this paper we are going to prove a corresponding result for an ultra-hyperbolic operator. 
It serves very well as a simple model operator that has a single \textit{non-degenerate} threshold 
of general signature in its continuous spectrum. 
An outline of the proof of Theorem~\ref{1606011129} will be given in Section~\ref{17021822}.

\subsection{Model operator}

Here we state a result on an ultra-hyperbolic operator on $\mathbb R^d$.
We note that the arguments of this subsection are similar and parallel 
to those of the previous subsection, but the notation is somewhat different. We have decided to repeat the arguments for the sake of clarity.

Let $\square=\square_{p,q}$ be the ultra-hyperbolic operator 
on $\mathbb R^d$ of signature $(p,q)$ with $p,q\ge 0$ and $d=p+q\ge 1$:
\begin{align}
\square=\partial_1^2+\dots+\partial_p^2-\partial_{p+1}^2-\dots-\partial_{p+q}^2.
\label{16060123}
\end{align}
On the Hilbert space $\mathcal H=L^2(\mathbb R^d)$ 
the operator $H_0=-\square$ has
a self-adjoint realization, denoted by $H_0$ again, 
with domain
\begin{align*}
\mathcal D(H_0)=\{u\in \mathcal H;\ \square u\in \mathcal H\text{ in the distributional sense}\},
\end{align*}
and it has a purely absolutely continuous spectrum:
\begin{align*}
\sigma(H_0)=\sigma_{\mathrm{ac}}(H_0)
=
\left\{
\begin{array}{ll}
[0,\infty)&\text{if }(p,q)=(d,0),\\
(-\infty,0]&\text{if }(p,q)=(0,d),\\
\mathbb R&\text{otherwise}.
\end{array}
\right.
\end{align*}
In fact, by the Fourier transform 
${\mathcal F}\colon \mathcal H \to\mathcal H$
and its inverse $\mathcal F^*\colon \mathcal H \to\mathcal H$ 
the operator $H_0$ is diagonalized as 
\begin{align*}
{\mathcal F}H_0\mathcal F^*=\Xi(\xi);\quad
\Xi(\xi)=\xi_1^2+\dots+\xi_p^2-\xi_{p+1}^2-\dots-\xi_{p+q}^2,
\end{align*}
where the right-hand side $\Xi(\xi)$ denotes 
a multiplication operator.
The only critical point $\xi=0$ of the function $\Xi(\xi)$ is of elliptic type
if $(p,q)=(d,0)$ or $(0,d)$, 
and is of hyperbolic type otherwise.
The critical value $0\in \sigma(H_0)$ is called a threshold,
and it is \textit{of elliptic} or \textit{hyperbolic type} if the associated 
critical point is of elliptic or hyperbolic type, respectively.

We are going to study the asymptotic behavior of the resolvent
\begin{align*}
R_0(z)=(H_0-z)^{-1}
\end{align*}
as the spectral parameter $z\in\mathbb C\setminus \sigma(H_0)$ approaches the threshold $0$.
We are particularly interested in its branching form. 
Let us utilize the following expression: For $u\in \mathcal S(\mathbb R^d)$
\begin{align}
(R_0(z)u)(x)
=
(2\pi)^{-d}\int_{\mathbb R^d}
\bigl[
\int_{\mathbb R^d}
\mathrm e^{\mathrm i(x-y)\xi}
(\Xi(\xi)-z)^{-1}u(y)\,\mathrm dy
\bigr]\mathrm d\xi.
\label{11.4.17.3.5}
\end{align}
Here the integral in \eqref{11.4.17.3.5} is understood by iterated integration,
though it is also possible to realize it by oscillatory integral.
As in the previous subsection we often split the indices $1,\dots, d$ into two groups consisting of 
the former $p$ and the latter $q$ ones,
and indicate relevant quantities by a prime and double primes, respectively,
e.g., we let 
\begin{align*}
\xi=(\xi',\xi'')\in\mathbb R^d=\mathbb R^p\oplus \mathbb R^q
\end{align*}
with $\mathbb R^0=\{0\}$, and write simply
\begin{align*}
\Xi(\xi)=\xi'^2-\xi''^2.
\end{align*}
If we define for $\gamma>0$
\begin{align}
k_\gamma(z,x)=(2\pi)^{-d}\int_{|\xi'|+|\xi''|<\gamma}
\frac{\mathrm e^{ix\xi}}{\xi'{}^2-\xi''{}^2-z}\,\mathrm d\xi,
\label{16043023}
\end{align}
then the expression \eqref{11.4.17.3.5} can be rewritten as the limiting convolution 
\begin{align*}
(R_0(z)u)(x)
=
\lim_{\gamma\to\infty}\int_{\mathbb R^d}k_\gamma(z,x-y)u(y)\,\mathrm dy
.
\end{align*}
Hence it suffices to investigate the convolution kernel \eqref{16043023}.
In fact, branching of the resolvent arises only from that 
of the kernel \eqref{16043023} with $\gamma>0$ arbitrarily small,
since the only critical point $\xi=0$ of $\Xi(\xi)=\xi'^2-\xi''^2$
is contained in the integration domain in \eqref{16043023}.

Let us introduce an entire function $E(w)$ by 
\begin{equation}
E(w)
=
\frac{\mathrm i^q}{2^d\pi^{d/2}}
\sum_{k=0}^\infty\frac{(-w/4)^k}{k!\Gamma(k+d/2)}
\quad 
\text{for }w\in\mathbb C
.
\label{E-def-ultra}
\end{equation}
We note that $E(w)$ can be expressed by 
the Bessel function of the first kind:
\begin{align*}
E(w)
=
\frac{\mathrm i^q}{2(2\pi)^{d/2}}\frac{J_{d/2-1}(w^{1/2})}{w^{(d/2-1)/2}}
;\quad 
J_\nu(z)=\sum_{k=0}^\infty\frac{(-1)^k(z/2)^{2k+\nu}}{k!\Gamma(k+\nu+1)}
.
\end{align*}
We also note that for any $z\in\mathbb C$
the function $E((x'^2-x''^2)z)$ in $x=(x',x'')\in\mathbb R^p\oplus\mathbb R^q$ 
satisfies the eigenequation 
\begin{align*}
(-\square-z)E\bigl((x'^2-x''^2)z\bigr)=0.
\end{align*}
We will see the last identity in Propositions~\ref{150803234} and \ref{160522193}.

Now we state a result for $\square$. 
Set $\Delta(\gamma^2)=\{z\in\mathbb C;\ |z|<\gamma^2\}$.

\begin{theorem}\label{150803236}
Let $p,q\ge 0$ with $p+q=d$, and $\gamma>0$. 
Then  the function $k_\gamma(z,x)$ defined by \eqref{16043023} for 
$(z,x)\in \mathbb C_+\times \mathbb C^d$
has the following expressions:
\begin{enumerate}
\item
If $d$ is odd,
there exists a function $\chi_\gamma(z,x)$ 
analytic in $(z,x)\in\Delta(\gamma^2)\times \mathbb C^d$
such that for $(z,x)\in (\mathbb C_+\cap \Delta(\gamma^2))\times\mathbb C^d$
\begin{align}
k_\gamma(z,x)
&
=\mathrm i\pi(\sqrt z)^{d-2}E\bigl((x'^2-x''^2)z\bigr)+\chi_\gamma(z,x);
\label{1606031643}
\end{align}
\item
If $d$ is even,
there exists a function  $\chi_\gamma(z,x)$ 
analytic in $(z,x)\in\Delta(\gamma^2)\times \mathbb C^d$
such that for $(z,x)\in (\mathbb C_+\cap \Delta(\gamma^2))\times\mathbb C^d$
\begin{align}
k_\gamma(z,x)
&
=-(\sqrt z)^{d-2}(\log z)E\bigl((x'^2-x''^2)z\bigr)
+\chi_\gamma(z,x).
\label{1606031644}
\end{align}
\end{enumerate}
\end{theorem}

Since $E((x'^2-x''^2)z)$ is entire in $z\in\mathbb C$,
the first terms on the right-hand sides of \eqref{1606031643} and \eqref{1606031644} 
are exactly the branching parts of $k_\gamma(z,x)$.
Thus the resolvent $R_0(z)$ has a square-root branching if $d$ is odd,
and a logarithm branching if $d$ is even.

We note that actually 
an explicit expression for the full kernel of the resolvent of an ultra-hyperbolic operator 
is already known~\cite{Brychkov-Prudnikov}.
It is written in terms of a Macdonald function, and, of course, its branching parts coincide with 
the above expressions.
However, in order to apply the arguments to the discrete Laplacian, 
we have to expand an integral expression more directly
without using a Macdonald function.

\section{Proofs for the model operator}\label{17021720}

The elliptic and hyperbolic cases are discussed separately.

\subsection{Elliptic threshold}\label{1605217}

We begin with the elliptic case. 
It suffices to consider the case $(p,q)=(d,0)$.
In the integral \eqref{16043023} 
let us introduce the spherical coordinates
$$\xi=\rho\omega,\quad  \rho\ge 0,\ \omega\in S^{d-1},$$
and rewrite \eqref{16043023} as 
\begin{align}
k_\gamma(z,x)
&=
\int_0^\gamma
\frac{\rho^{d-1}e(\rho x)}{\rho^2-z}\,\mathrm d\rho,
\label{1508011300}
\end{align}
where we have set in general for $\zeta\in \mathbb C^d$
\begin{align}
e(\zeta)&=(2\pi)^{-d}\int_{S^{d-1}} \mathrm e^{\mathrm i\zeta\omega}\,\mathrm dS(\omega).
\label{11.4.15.5.6}
\end{align}
For $d=1$ the expression \eqref{11.4.15.5.6} reads
\begin{align*}
e(\zeta)=
(2\pi)^{-1}\bigl(\mathrm e^{\mathrm i\zeta}+\mathrm e^{-\mathrm i\zeta}\bigr)
=\pi^{-1}\cos \zeta.
\end{align*}
The function $e(\zeta)$ has the following properties:
\begin{proposition}\label{150803234}
The function $e(\zeta)$ satisfies the identities
\begin{align*}
e(\zeta)=2E(\zeta^2);\quad \zeta^2=\zeta_1^2+\dots+\zeta^d,
\end{align*}
and 
\begin{align*}
-(\partial_1^2+\dots+\partial_d^2)e(\zeta)=e(\zeta).
\end{align*}
\end{proposition}
\begin{proof}
The latter identity is obvious due to the expression \eqref{11.4.15.5.6}.
To prove the former one it suffices to show that the function $e(\zeta)$ has the Taylor expansion
\begin{align}
e(\zeta)=\frac{2}{2^d\pi^{d/2} }\sum_{\alpha\in \mathbb Z_+^d}
\frac{(-1)^{|\alpha|}(\zeta/2)^{2\alpha}}{\alpha!\Gamma(|\alpha|+d/2)}
.
\label{1605216}
\end{align}
The expansion \eqref{1605216} is obvious for $d=1$, and we may let $d\ge 2$.
By Euler's formula and symmetry we can write 
\begin{align*}
e(\zeta)
&=\frac1{(2\pi)^d}\int_{S^{d-1}}\prod_{j=1}^d\cos (\zeta_j\omega_j)\,\mathrm dS(\omega)
\\&
=
\frac1{(2\pi)^d}\sum_{\alpha\in\mathbb Z_+^d}
\frac{(-1)^{|\alpha|}}{(2\alpha)!}\zeta^{2\alpha}
\int_{S^{d-1}} \omega^{2\alpha}\,\mathrm dS(\omega).
\end{align*}
Let us compute the integral on the right-hand side using the 
spherical coordinates 
$\theta=(\theta_1,\dots,\theta_{d-1})\in [0,\pi]^{d-2}\times [0,2\pi)$
with 
\begin{align*}
\omega_j&=\left\{
\begin{array}{ll}
\cos\theta_1 &\mbox{for }j=1,\\
\sin\theta_1\cdots\sin\theta_{j-1}\cos\theta_j&\mbox{for } 2\le j\le d-1,\\
\sin\theta_1\cdots\sin\theta_{d-1}&\mbox{for } j=d,
\end{array}
\right.
\end{align*}
for which we have 
\begin{align*}
\mathrm dS(\omega)=\sin^{d-2}\theta_1\cdots\sin\theta_{d-2}
\,\mathrm d\theta_1\cdots\mathrm d\theta_{d-1}
.
\end{align*}
After some computations employing the beta function we shall obtain
\begin{align}
\begin{split}
\int_{S^{d-1}} \omega^{2\alpha}\,\mathrm dS(\omega)
&=\frac{2\prod_{j=1}^{d}\Gamma(\alpha_j+1/2)}{\Gamma(|\alpha|+d/2)}
=\frac{2\pi^{d/2}(2\alpha)!}{4^{|\alpha|}\alpha!\Gamma(|\alpha|+d/2)},
\end{split}
\label{1606011}
\end{align}
and hence the asserted expansion \eqref{1605216} is verified.
\end{proof}

Now we prove Theorem~\ref{150803236} for the elliptic case.

\begin{proof}[Proof of Theorem~\ref{150803236} for the elliptic case]
\textit{1.}\quad 
Let the dimension $d$ be odd, and $(p,q)=(d,0)$.
We compute the integral \eqref{1508011300}.
Since the integrand is even in $\rho$, we can write
\begin{align*}
\begin{split}
k_\gamma(z,x)
&{}
=
-\tfrac12\int_\gamma^{-\gamma}\frac{\rho^{d-1}e(\rho x)}{(\rho-\sqrt z)(\rho+\sqrt z)}\,\mathrm d\rho.
\end{split}
\end{align*}
Let us change the contour of integration to the contour
$$\widetilde\Gamma(\gamma)=\{\gamma\mathrm e^{\mathrm i\theta}\in\mathbb C;\ \theta\in[0,\pi]\},$$
using the calculus of residues. Then we obtain
\begin{align*}
\begin{split}
k_\gamma(z,x)
&{}
=
\tfrac{\mathrm i\pi}2 (\sqrt z)^{d-2}e(\sqrt z x)
+\chi_\gamma(z,x)
\end{split}
\end{align*}
with 
\begin{align*}
\chi_\gamma(z,x)
=-\tfrac12\int_{\widetilde\Gamma(\gamma)}\frac{\rho^{d-1}e(\rho x)}{\rho^2-z}\,\mathrm d\rho.
\end{align*}
The above function $\chi_\gamma(z,x)$
is obviously analytic in $(z,x)\in\Delta(\gamma^2)\times \mathbb C^d$. Hence we are done.

\smallskip
\noindent
\textit{2.}\quad 
We next let the dimension $d$ be even, and $(p,q)=(d,0)$.
In the integral \eqref{1508011300}
let us change variables from $\rho$ to $\lambda=\rho^2$:
\begin{align*}
k_\gamma(z,x)
&=
\tfrac12\int_0^{\gamma^2}\frac{(\sqrt\lambda)^{d-2}e(\sqrt{\lambda} x)}{\lambda-z}\,\mathrm d\lambda.
\end{align*}
If we set the difference quotient
\begin{align}
\widetilde e(\lambda,z,x)=
\frac{(\sqrt{\lambda})^{d-2}e(\sqrt{\lambda}x)
-(\sqrt z)^{d-2}e(\sqrt z x)}{\lambda-z},
\label{16052317}
\end{align}
then we can compute 
\begin{align*}
k_\gamma(z,x)
&=
\tfrac12\int_0^{\gamma^2}\frac{(\sqrt z)^{d-2}e(\sqrt{z} x)}{\lambda-z}\,\mathrm d\lambda
+\tfrac12\int_0^{\gamma^2}\widetilde e(\lambda,z,x)\,\mathrm d\lambda
\\
&=
-\tfrac12(\sqrt z)^{d-2}e(\sqrt{z} x)\log z
+\chi_\gamma(z,x)
\end{align*}
with 
\begin{align}
\chi_\gamma(z,x)=
\tfrac12(\sqrt z)^{d-2}e(\sqrt{z} x)\bigl(\log(\gamma^2-z)+\mathrm i\pi\bigr)
+\tfrac12\int_0^{\gamma^2}\widetilde e(\lambda,z,x)\,\mathrm d\lambda.
\label{16052318}
\end{align}
Since $d$ is even and $e(\zeta)$ is even and entire in $\zeta\in\mathbb C^d$, 
the difference quotient $\widetilde e(\lambda,z,x)$ from \eqref{16052317}
is entire in $(\lambda,z,x)\in\mathbb C\times\mathbb C\times\mathbb C^d$
and even in $x\in\mathbb C^d$.
Hence the function $\chi_\gamma(z,x)$ from \eqref{16052318} is analytic
in $(z,x)\in\Delta(\gamma^2)\times \mathbb C^d$.
Hence we are done.
\end{proof}

\subsection{Hyperbolic threshold}\label{160514}

Now we consider the hyperbolic case: $p,q\ge 1$.
Introduce the \textit{split spherical coordinates} 
\begin{align}
\xi'=\rho'\omega',\ 
\xi''=\rho''\omega'';\quad
\rho',\rho''\ge 0,\ 
\omega'\in S^{p-1},\ 
\omega''\in S^{q-1},
\label{15081323}
\end{align}
and rewrite \eqref{16043023} as 
\begin{align}
k_\gamma(z,x)
&{}=\int_{\genfrac{}{}{0pt}{}{\rho',\rho''\ge 0,}{\rho'+\rho''<\gamma}} 
\frac{\rho'^{p-1}\rho''^{q-1}e'(\rho'x')e''(\rho''x'')}{\rho'^2-\rho''^2-z}
\,\mathrm d\rho'\mathrm d\rho'',
\label{11.12.12.14.52}
\end{align}
where 
\begin{align}
\begin{split}
e'(\zeta')&{}
=(2\pi)^{-p}\int_{S^{p-1}}
\mathrm e^{\mathrm i\zeta'\omega'}\,\mathrm dS(\omega'),
\\
e''(\zeta'')
&{}
=(2\pi)^{-q}\int_{S^{q-1}}
\mathrm e^{\mathrm i\zeta''\omega''} \,\mathrm dS(\omega'').
\end{split}
\label{11.12.7.22.52}
\end{align}
Split the integration region of \eqref{11.12.12.14.52} into two parts:
\begin{align}
\begin{split}
k_\gamma(z,x)
&{}=\int_{\genfrac{}{}{0pt}{}{\rho'\ge \rho''\ge 0,}{\rho'+\rho''<\gamma}} 
\frac{\rho'^{p-1}\rho''^{q-1}e'(\rho'x')e''(\rho''x'')}{\rho'^2-\rho''^2-z}
\,\mathrm d\rho'\mathrm d\rho''
\\&\phantom{={}}
+\int_{\genfrac{}{}{0pt}{}{\rho''>\rho'\ge 0,}{\rho'+\rho''<\gamma}} 
\frac{\rho'^{p-1}\rho''^{q-1}e'(\rho'x')e''(\rho''x'')}{\rho'^2-\rho''^2-z}
\,\mathrm d\rho'\mathrm d\rho'',
\end{split}
\label{160501}
\end{align}
and further change the variables in the above two integrals by 
\begin{align*}
\rho'=\tfrac12{\tau(\sigma\pm\sigma^{-1})},
\quad
\rho''=\tfrac12{\tau(\sigma\mp\sigma^{-1})},
\end{align*}
respectively.
Then, introducing the functions
\begin{align}
\begin{split}
f_{\pm}(\sigma,\zeta)&=
\frac{(\sigma\pm\sigma^{-1})^{p-1}(\sigma\mp\sigma^{-1})^{q-1}}{2^{d-2}}
\\&\phantom{{}={}}
\cdot
e'\bigl(\tfrac12(\sigma\pm\sigma^{-1})\zeta'\bigr)
e''\bigl(\tfrac12(\sigma\mp\sigma^{-1})\zeta''\bigr)
\end{split}
\label{11.12.12.22.3}
\end{align}
and 
\begin{align}
g_{\pm,\gamma}(\tau,x)&=\tau^{d-2}\int_1^{\gamma/\tau}
\frac{ f_\pm(\sigma,\tau x)}{\sigma}\,\mathrm d\sigma
,
\label{1605082}
\end{align}
we can rewrite \eqref{160501} as
\begin{align}
\begin{split}
k_\gamma(z,x)
&
=\int_0^\gamma\frac{\tau g_{+,\gamma}(\tau,x)}{\tau^2-z}\,\mathrm d\tau
-\int_0^\gamma\frac{\tau g_{-,\gamma}(\tau,x)}{\tau^2+z}\,\mathrm d\tau.
\end{split}\label{11.12.12.15.4}
\end{align}

Obviously, the functions $f_{\pm}(\sigma,\zeta)$ from \eqref{11.12.12.22.3}
are continued as single-valued analytic function in $(\sigma,\zeta)
\in (\mathbb C\setminus\{0\})\times \mathbb C^d$.
However, the analytic continuations of the functions $g_{\pm,\gamma}(\tau,x)$ from \eqref{1605082}
are not necessarily single-valued.
Such possible branchings originate from 
residues at $\sigma=0$ of the integrands $f_\pm(\sigma,\zeta)/\sigma$. 
Let us give them special names:
\begin{align}
\phi_\pm(\zeta)
&
=\frac1{2\pi\mathrm i}\int_{|\sigma|=1}
\frac{f_\pm(\sigma,\zeta)}{\sigma}\,\mathrm d\sigma
.
\label{11.12.11.13.40}
\end{align}
Let us also define auxiliary functions
\begin{align}
\psi_\pm(\zeta)
&=\int_{\Gamma(1)}\frac{f_\pm(\sigma,\zeta)-\phi_\pm(\zeta)}{\sigma}\,\mathrm d\sigma
\label{16051123}
\end{align}
with
\begin{align*}
\Gamma(1)&=\{\mathrm e^{\mathrm i\theta}\in\mathbb C;\ \theta\in[0,\pi/2]\}.
\end{align*}
These functions describe the branching part of $k_\gamma(z,x)$.
We state basic properties of them as a proposition. We use the function $E(w)$ defined in \eqref{E-def-ultra}.
\begin{proposition}\label{160522193}
The functions $\phi_\pm(\zeta)$ and $\psi_\pm(\zeta)$ satisfy
$$
\phi_+(\zeta)
=\frac{4}{\mathrm i\pi}
E(\zeta'^2-\zeta''^2)
,\quad
\phi_-(\zeta)
=\frac{4\mathrm i^{p-q}}{\mathrm i\pi}
E(\zeta''^2-\zeta'^2)
,\quad 
\psi_\pm(\zeta)=0$$
if $(p,q)$ is odd-odd, and 
$$
\phi_\pm(\zeta)=0
,\quad 
\psi_+(\zeta)=2E(\zeta'^2-\zeta''^2)
,\quad
\psi_-(\zeta)=2\mathrm i^{p-q}E(\zeta''^2-\zeta'^2)
$$
otherwise. In addition, they also satisfy 
\begin{align}
-(\partial_1^2+\dots+\partial_p^2-\partial_{p+1}^2-\dots-\partial_{p+q}^2)\phi_\pm(\zeta)&=\phi_\pm(\zeta)
,\label{17021308}
\\
-(\partial_1^2+\dots+\partial_p^2-\partial_{p+1}^2-\dots-\partial_{p+q}^2)\psi_\pm(\zeta)&=\psi_\pm(\zeta)
.
\end{align}
\end{proposition}
\begin{proof}
We first compute $\phi_+(\zeta)$.
Let $\sigma=\mathrm e^{\mathrm i\theta}$, $\theta\in(-\pi,\pi]$, in \eqref{16051123}.
Then, noting \eqref{11.12.12.22.3}, we have 
\begin{align*}
\phi_+(\zeta)
&
=\frac{\mathrm i^{q-1}}{2\pi}\int_{-\pi}^{\pi}
\cos^{p-1}\theta
\sin^{q-1}\theta
e'(\zeta'\cos\theta)
e''(\mathrm i\zeta''\sin\theta)
\,\mathrm d\theta
.
\end{align*}
By symmetry of the integrand it easily follows 
that 
\begin{align}
\begin{split}
\phi_+(\zeta)
&
=\frac{2\mathrm i^{q-1}}{\pi}\int_0^{\pi/2}
\cos^{p-1}\theta
\sin^{q-1}\theta
e'(\zeta'\cos\theta)
e''(\mathrm i\zeta''\sin\theta)
\,\mathrm d\theta
\\&
=
\frac{2\mathrm i^{q-1}}{\pi(2\pi)^d}
\int_{[0,\pi/2]\times S^{p-1}\times S^{q-1}}
\exp\bigl(\mathrm i\zeta'\omega'\cos\theta+\mathrm i\zeta''\omega''\sin\theta\bigr)
\\&\phantom{{}=\frac{2\mathrm i^{q-1}}{\pi(2\pi)^d}\int}
\cos^{p-1}\theta
\sin^{q-1}\theta
\,\mathrm d\theta\mathrm dS(\omega')\mathrm dS(\omega'')
\end{split}
\label{170213}
\end{align}
if $(p,q)$ is odd-odd, 
and $\phi_+(\zeta)=0$ otherwise.
Hence it remains to compute \eqref{170213} for $(p,q)$ odd-odd.
In \eqref{170213} we change variables to 
$$\omega
=(\omega'\cos\theta,\omega''\sin\theta)\in S^{d-1}.$$
Then, noting that 
\begin{align*}
\mathrm dS(\omega)
=\cos^{p-1}\theta
\sin^{q-1}\theta
\,\mathrm d\theta\mathrm dS(\omega')\mathrm dS(\omega''),
\end{align*}
we can write
\begin{align}
\phi_+(\zeta)
=\frac{2\mathrm i^{q-1}}{\pi}
e(\zeta',\mathrm i\zeta'')
=-\frac{4\mathrm i}{\pi}
E(\zeta'^2-\zeta''^2)
.
\label{1702130842}
\end{align}
The identity \eqref{17021308} for $\phi_+(\zeta)$ is clear from the middle expression from
\eqref{1702130842}.
Hence we are done with $\phi_+(\zeta)$

As for $\phi_-(\zeta)$, we let 
$\sigma=\mathrm e^{\mathrm i\theta}$, $\theta\in(-\pi,\pi]$, in \eqref{16051123},
and write 
\begin{align*}
\phi_-(\zeta)
&
=\frac{\mathrm i^{p-1}}{2\pi}\int_{-\pi}^{\pi}
\sin^{p-1}\theta
\cos^{q-1}\theta
e'(\mathrm i\zeta'\sin\theta)
e''(\zeta''\cos\theta)
\,\mathrm d\theta.
\end{align*}
Then we can proceed as in the above argument. We omit the details.

Next, let us compute $\psi_+(\zeta)$.
By letting 
$\sigma=\mathrm e^{\mathrm i\theta}$, $\theta\in[0,\pi/2]$, 
in \eqref{16051123}
we can write 
\begin{align*}
\psi_+(\zeta)
&=
\mathrm i^q\int_0^{\pi/2}
\cos^{p-1}\theta
\sin^{q-1}\theta
e'(\zeta'\cos\theta)
e''(\mathrm i\zeta''\sin\theta)
\,\mathrm d\theta
-\frac{\mathrm i\pi}2\phi_+(\zeta).
\end{align*}
Then, noting \eqref{170213} for $(p,q)$ odd-odd, 
we obtain 
$\psi_+(\zeta)=0$
if $(p,q)$ is odd-odd, and 
\begin{align*}
\psi_+(\zeta)
&=
\mathrm i^q\int_0^{\pi/2}
\cos^{p-1}\theta
\sin^{q-1}\theta
e'(\zeta'\cos\theta)
e''(\mathrm i\zeta''\sin\theta)
\,\mathrm d\theta
\end{align*}
otherwise. Now we can argue similarly to $\phi_+(\zeta)$,
and are done with $\psi_+(\zeta)$.

We can discuss $\psi_-(\zeta)$ similarly, and omit the details.
We are done.
\end{proof}

\subsubsection{Odd-even or even-odd signature}\label{17021319}

Here we prove Theorem~\ref{150803236} for $(p,q)$ odd-even or even-odd.
In this case the functions 
$g_\pm(\tau,x)$ from \eqref{1605082}
are independent of choice of contours of integrations,
since residues of the integrands at $\sigma=0$ are $0$ by Proposition~\ref{160522193}.
The functions $\psi_\pm(\zeta)$ from \eqref{16051123} appear naturally 
as the odd parts of $g_{\pm,\gamma}(\tau,x)$. 
We shall make use of such odd and even decompositions to 
implement the integrations in \eqref{11.12.12.15.4}.

\begin{lemma}\label{1605123}
Let $(p,q)$ be odd-even or even-odd, and $\gamma>0$.
Then there exists an entire function $h_\gamma(w,x)$
in $(w,x)\in \mathbb C\times \mathbb C^d$
such that 
the functions $g_\pm(\tau,x)$ defined 
by \eqref{1605082} for $(\tau,x)\in(\mathbb C\setminus\{0\})\times \mathbb C^d$ 
have the expressions
\begin{align}
g_{\pm,\gamma}(\tau,x)&
=\tfrac12
\bigl[
\tau^{d-2}\psi_\pm(\tau x)+(\mathrm i\tau)^{d-2}\psi_\mp(\mathrm i\tau x)
\bigr]
+h_{\gamma}(\pm\tau^2,x)
.\label{16050815}
\end{align}
In particular, $g_{\pm,\gamma}(\tau,x)$ extend to be entire 
in $(\tau,x)\in \mathbb C\times \mathbb C^d$.
\end{lemma}
\begin{proof}
We first note that, since $f_\pm(\sigma,\zeta)/\sigma$ have residue $0$ at $\sigma=0$,
the integrals in this proof do not depend on contours of integrations.
Using the identities
\begin{align}
f_\pm(-\sigma,\zeta)=-f_\pm(\sigma,\zeta),\quad 
f_\pm(\sigma,-\zeta)=f_\pm(\sigma,\zeta),
\label{1702131630}
\end{align}
we can verify that for any $(\tau,x)\in (\mathbb C\setminus\{0\})\times \mathbb C^d$
\begin{align}
g_{\pm,\gamma}(\tau,x)
-g_{\pm,\gamma}(-\tau,x)
&=\tau^{d-2}
\int_1^{-1}
\frac{ f_\pm(\sigma,\tau x)}{\sigma}\,\mathrm d\sigma,
\label{17021316}
\\
g_{\pm,\gamma}(\tau,x)
+g_{\pm,\gamma}(-\tau,x)
&=
\tau^{d-2}\int_{-\gamma/\tau}^{\gamma/\tau}
\frac{ f_\pm(\sigma,\tau x)}{\sigma}\,\mathrm d\sigma
.
\label{1702131617}
\end{align}
Moreover, if we use
\begin{align}
f_\pm(\mathrm i\sigma,\zeta)=\mathrm i^{d-2}f_\mp(\sigma,\mathrm i\zeta),
\label{1702131631}
\end{align}
we can rewrite \eqref{17021316} as 
\begin{align*}
g_{\pm,\gamma}(\tau,x)
-g_{\pm,\gamma}(-\tau,x)
=\tau^{d-2}\psi_\pm(\tau x)+(\mathrm i\tau)^{d-2}\psi_\mp(\mathrm i\tau x).
\end{align*}
On the other hand, let us rewrite \eqref{1702131617} by change of variables as 
\begin{align}
g_{\pm,\gamma}(\tau,x)+g_{\pm,\gamma}(-\tau,x)
&=\int_{-\gamma}^{\gamma}
\frac{\tau^{d-2}f_\pm(\sigma/\tau,\tau x)}{\sigma}\,\mathrm d\sigma.
\label{160511233}
\end{align}
It is not difficult to see that the right-hand side of \eqref{160511233}
extends to be entire, 
since the integrands combined are analytic in 
$(\sigma,\tau,x)\in (\mathbb C\setminus\{0\})\times\mathbb C\times \mathbb C^d$:
Apparent singularities at $\tau=0$ are in fact removable 
by noting the expression \eqref{11.12.12.22.3}.
Since \eqref{160511233} is even in $\tau\in\mathbb C$, we may set 
$$h_\gamma(\tau^2,x)=\tfrac12\int_{-\gamma}^{\gamma}
\frac{\tau^{d-2}f_+(\sigma/\tau,\tau x)}{\sigma}\,\mathrm d\sigma.$$
If we use \eqref{1702131630} and \eqref{1702131631} again, 
we can verify 
$$h_\gamma(-\tau^2,x)=\tfrac12\int_{-\gamma}^{\gamma}
\frac{\tau^{d-2}f_-(\sigma/\tau,\tau x)}{\sigma}\,\mathrm d\sigma.$$
Hence we are done.
\end{proof}

We now prove Theorem~\ref{150803236} for $(p,q)$ odd-even or even-odd.

\begin{proof}[Proof of Theorem~\ref{150803236} for $(p,q)$ odd-even or even-odd]
In this proof for notational simplicity we set 
\begin{align*}
z=\kappa^2,\quad 
\kappa\in\mathbb C_{++}=\{w\in \mathbb C;\ \mathop{\mathrm{Re}}w>0,\ \mathop{\mathrm{Im}}w>0\}.
\end{align*}
Let us substitute the odd and even decompositions \eqref{16050815}
into \eqref{11.12.12.15.4}:
\begin{align}
\begin{split}
k_\gamma(z,x)
&
=
\tfrac12
\int_0^\gamma\frac{\tau^{d-1}\psi_+(\tau x)}{\tau^2-z}\,\mathrm d\tau
+\tfrac12\int_0^\gamma\frac{\tau(\mathrm i\tau)^{d-2}\psi_-(\mathrm i\tau x)}{\tau^2-z}\,\mathrm d\tau
\\&\phantom{{}={}}
-\tfrac12\int_0^\gamma\frac{\tau^{d-1}\psi_-(\tau x)}{\tau^2+z}\,\mathrm d\tau
-\tfrac12\int_0^\gamma\frac{\tau(\mathrm i\tau)^{d-2}\psi_+(\mathrm i\tau x)}{\tau^2+z}\,\mathrm d\tau
\\&\phantom{{}={}}
+\int_0^\gamma\frac{\tau h_{\gamma}(\tau^2,x)}{\tau^2-z}\,\mathrm d\tau
-\int_0^\gamma\frac{\tau h_{\gamma}(-\tau^2,x)}{\tau^2+z}\,\mathrm d\tau
\\
&=:I_1+I_2-I_3-I_4+I_5-I_6,
\end{split}
\label{1605081522}
\end{align}
where all the contours of integrations are set on the real axis.
For the first integral $I_1$ in the brackets of \eqref{1605081522} 
we note that the integrand is even in $\tau$,
so that the interval of integration can be symmetrized.
Then, using Cauchy's integral formula, we have 
\begin{align*}
\begin{split}
I_1
&=
\tfrac14\int_{-\gamma}^\gamma\frac{\tau^{d-1}\psi_+(\tau x)}{\tau^2-z}\,\mathrm d\tau
=\frac{\mathrm i\pi}4 \kappa^{d-2}\psi_+(\kappa x)-
\tfrac14\int_{\widetilde\Gamma(\gamma)}\frac{\tau^{d-1}\psi_+(\tau x)}{\tau^2-z}\,\mathrm d\tau.
\end{split}
\end{align*}
Similarly for the second to fourth integrals from \eqref{1605081522},
\begin{align*}
\begin{split}
I_2
&
=\frac{\mathrm i\pi}4 (\mathrm i\kappa)^{d-2}\psi_-(\mathrm i\kappa x)
-\tfrac14\int_{\widetilde\Gamma(\gamma)}\frac{\tau(\mathrm i\tau)^{d-2}\psi_-(\mathrm i\tau x)}{\tau^2-z}\,\mathrm d\tau,
\end{split}
\\
\begin{split}
I_3
&
=\frac{\mathrm i\pi}4 (\mathrm i\kappa)^{d-2}\psi_-(\mathrm i\kappa x)
-\tfrac14\int_{\widetilde\Gamma(\gamma)}\frac{\tau^{d-1}\psi_-(\tau x)}{\tau^2+z}\,\mathrm d\tau,
\end{split}
\\
\begin{split}
I_4&
=\frac{\mathrm i\pi}4 (-\kappa)^{d-2}\psi_+(-\kappa x)
-\tfrac14\int_{\widetilde\Gamma(\gamma)}\frac{\tau(\mathrm i\tau)^{d-2}\psi_+(\mathrm i\tau x)}{\tau^2+z}\,\mathrm d\tau,
\end{split}
\end{align*}
so that by summing up and changing variables
\begin{align*}
\begin{split}
&I_1+I_2-I_3-I_4
\\&
=\frac{\mathrm i\pi}2\kappa^{d-2}\psi_+(\kappa x)
-\tfrac14\int_{\widetilde\Gamma(\gamma)}\frac{\tau^{d-1}\psi_+(\tau x)}{\tau^2-z}\,\mathrm d\tau
-\tfrac14\int_{\mathrm i\widetilde\Gamma(\gamma)}
\frac{\tau^{d-1}\psi_-(\tau x)}{\tau^2+z}\,\mathrm d\tau
\\&\phantom{{}={}}
+\tfrac14\int_{\widetilde\Gamma(\gamma)}\frac{\tau^{d-1}\psi_-(\tau x)}{\tau^2+z}\,\mathrm d\tau
+\tfrac14\int_{\mathrm i\widetilde\Gamma(\gamma)}
\frac{\tau^{d-1}\psi_+(\tau x)}{\tau^2-z}\,\mathrm d\tau
\\&
=\frac{\mathrm i\pi}2\kappa^{d-2}\psi_+(\kappa x)
-\tfrac12\int_{\Gamma(\gamma)}\frac{\tau^{d-1}\psi_+(\tau x)}{\tau^2-z}\,\mathrm d\tau
+\tfrac12\int_{\Gamma(\gamma)}\frac{\tau^{d-1}\psi_-(\tau x)}{\tau^2+z}\,\mathrm d\tau
,
\end{split}
\end{align*}
where $\mathrm i\widetilde\Gamma(\gamma)=\{\mathrm iw;\ w\in\widetilde\Gamma(\gamma)\}$.
For the fifth and sixth integrals from \eqref{1605081522} 
we change the variables to $\tau^2=\lambda$ 
and $\tau^2=-\lambda$, respectively, 
and combine them as 
\begin{align*}
\begin{split}
I_5-I_6&
=
\tfrac12
\int_0^{\gamma^2}\frac{h_{\gamma}(\lambda,x)}{\lambda-z}\,\mathrm d\lambda
+\tfrac12
\int_0^{-\gamma^2}\frac{h_{\gamma}(\lambda,x)}{-\lambda+z}\,\mathrm d\lambda
\\&
=
\tfrac12\int_{\mathrm i^2\widetilde\Gamma(\gamma^2)}\frac{h_{\gamma}(\lambda,x)}{\lambda-z}\,\mathrm d\lambda,
\end{split}
\end{align*}
where $\mathrm i^2\widetilde\Gamma(\gamma^2)
=\{\mathrm i^2w;\ w\in\widetilde\Gamma(\gamma^2)\}$.
Hence, 
if we set 
\begin{align*}
\begin{split}
\chi_\gamma(z,x)
&=-\tfrac12\int_{\Gamma(\gamma)}\frac{\tau^{d-1}\psi_+(\tau x)}{\tau^2-z}\,\mathrm d\tau
+\tfrac12\int_{\Gamma(\gamma)}\frac{\tau^{d-1}\psi_-(\tau x)}{\tau^2+z}\,\mathrm d\tau
\\&\phantom{{}={}}
+\tfrac12\int_{\mathrm i^2\widetilde\Gamma(\gamma^2)}
\frac{h_{\gamma}(\lambda,x)}{\lambda-z}\,\mathrm d\lambda,
\end{split}
\end{align*}
which is obviously analytic in $(z,x)\in \Delta(\gamma^2)\times\mathbb C^d$, 
then we obtain
\begin{align*}
k_\gamma(z,x)
=
\frac{\mathrm i\pi}2\kappa^{d-2}\psi_+(\kappa x)
+\chi_\gamma(z,x)
\end{align*}
This implies the assertion.
\end{proof}

\subsubsection{Even-even signature}\label{1605149}

Next, we prove Theorem~\ref{150803236} for $(p,q)$ even-even.
In this case, as in Section~\ref{17021319}, by Proposition~\ref{160522193}
the functions $g_\pm(\tau,x)$
extend analytically as single-valued functions.
However, as stated in Lemma~\ref{1605222022} below, 
here we decompose $g_\pm(\tau,x)$ depending on
symmetry under a quarter-rotation, not under half-rotation like in Section~\ref{17021319}.

\begin{lemma}\label{1605222022}
Let $(p,q)$ be even-even, and $\gamma>0$.
Then there exists an entire function $h_\gamma(w,x)$ in $(w,x)\in \mathbb C\times\mathbb C^d$
such that the functions $g_{\pm,\gamma}(\tau,x)$ defined by 
\eqref{1605082} for $(\tau,x)\in(\mathbb C\setminus \{0\})\times\mathbb C^d$
have the expressions
\begin{align}
g_{\pm,\gamma}(\tau,x)
&=
\tfrac12\tau^{d-2}\psi_\pm(\tau x)+ h_{\gamma}(\pm\tau^2,x).
\label{16051223}
\end{align}
In particular, $g_{\pm,\gamma}(\tau,x)$ extend to be entire in $(\tau,x)\in\mathbb C\times\mathbb C^d$.
\end{lemma}
\begin{proof}
Let us change variables in \eqref{1605082}, and decompose
\begin{align}
g_{\pm,\gamma}(\tau,x)
=
\tfrac12\tau^{d-2}\psi_\pm(\tau x)
+
\tau^{d-2}\biggl[
\int_\tau^\gamma\frac{f_\pm(\sigma/\tau,\tau x)}{\sigma}
\,\mathrm d\sigma
-\tfrac12\psi_\pm(\tau x)
\biggr].
\label{16051423b}
\end{align}
As in the proof of Lemma~\ref{1605123}, 
the apparent singularities at $\tau=0$ of the latter terms on the right-hand side of 
\eqref{16051423b}
are in fact removable by noting the expressions \eqref{11.12.12.22.3},
and they are entire in $(\tau,x)\in \mathbb C\times \mathbb C^d$.
Moreover, by \eqref{1702131630}, \eqref{1702131631} and Proposition~\ref{160522193} 
they are even in $\tau\in \mathbb C$.
Hence we may in particular write 
\begin{align*}
h_{\gamma}(\tau^2,x)
&=
\tau^{d-2}\biggl[
\int_\tau^\gamma\frac{f_+(\sigma/\tau,\tau x)}{\sigma}
\,\mathrm d\sigma
-\tfrac12\psi_+(\tau x)\biggr]
,
\end{align*}
where $h_\gamma(w,x)$ is 
an entire function in $(w,x)\in \mathbb C\times\mathbb C^d$.
Then it follows by \eqref{1702131630}, \eqref{1702131631}, Proposition~\ref{160522193}
and \eqref{16051123}
that 
\begin{align*}
h_{\gamma}(-\tau^2,x)
&=
(\mathrm i\tau)^{d-2}\int_{\mathrm i\tau}^{\gamma}
\frac{f_+\bigl(\sigma/(\mathrm i\tau),\mathrm i\tau x\bigr)}{\sigma}\,\mathrm d\sigma
-\tfrac12(\mathrm i\tau)^{d-2}\psi_+(\mathrm i\tau x)
\\&=
\tau^{d-2}\int_{\mathrm i\tau}^{\gamma}
\frac{f_-(\sigma/\tau,\tau x)}{\sigma}\,\mathrm d\sigma
+\tfrac12\tau^{d-2}\psi_-(\tau x)
\\&=
\tau^{d-2}\int_{\mathrm i}^{\gamma/\tau}
\frac{f_-(\sigma,\tau x)}{\sigma}\,\mathrm d\sigma
+\tfrac12\tau^{d-2}\psi_-(\tau x)
\\&=
\tau^{d-2}\int_1^{\gamma/\tau}
\frac{f_-(\sigma,\tau x)}{\sigma}\,\mathrm d\sigma
-\tfrac12\tau^{d-2}\psi_-(\tau x)
\\&=
\tau^{d-2}\biggl[
\int_\tau^\gamma\frac{f_-(\sigma/\tau,\tau x)}{\sigma}
\,\mathrm d\sigma
-\tfrac12\psi_-(\tau x)\biggr].
\end{align*}
This verifies the assertion.
\end{proof}

\begin{proof}[Proof of Theorem~\ref{150803236} for $(p,q)$ even-even]
Substitute the identities \eqref{16051223} into the expression \eqref{11.12.12.15.4},
and change variables of the integrations. 
Then, also using Proposition~\ref{160522193}, we have
\begin{align*}
\begin{split}
k_\gamma(z,x)
&
=
\tfrac12\int_0^\gamma\frac{\tau^{d-1}\psi_+(\tau x)}{\tau^2-z}\,\mathrm d\tau
+\tfrac12\int_0^\gamma\frac{\tau h_{\gamma}(\tau^2,x)}{\tau^2-z}\,\mathrm d\tau
\\&\phantom{{}={}}
-\int_0^\gamma\frac{\tau^{d-1}\psi_-(\tau x)}{\tau^2+z}\,\mathrm d\tau
-\int_0^\gamma\frac{\tau h_{\gamma}(-\tau^2,x)}{\tau^2+z}\,\mathrm d\tau.
\\&
=
\tfrac14\int_0^{\gamma^2}\frac{(\sqrt{\lambda})^{d-2}\psi_+(\sqrt{\lambda}x)}{\lambda-z}\,\mathrm d\lambda
-\tfrac14\int_{-\gamma^2}^0\frac{(\sqrt{\lambda})^{d-2}\psi_+(\sqrt{\lambda}x)}{\lambda-z}\,\mathrm d\lambda
\\&\phantom{{}={}}
+\tfrac12\int_{-\gamma^2}^{\gamma^2}\frac{h_{\gamma}(\lambda,x)}{\lambda-z}\,\mathrm d\lambda.
\end{split}
\end{align*}
Now, let us set
for $(z,x)\in \mathbb C_+\times\mathbb C^d$  
\begin{align*}
\begin{split}
\eta_\gamma(z,x)
&=
\tfrac14\int_0^{\gamma^2}\widetilde{\psi}(\tau,z,x)\,\mathrm d\tau
-\tfrac14\int_{-\gamma^2}^0\widetilde{\psi}(\tau,z,x)\,\mathrm d\tau
+\tfrac12\int_{-\gamma^2}^{\gamma^2}\frac{h_{\gamma}(\lambda,x)}{\lambda-z}\,\mathrm d\lambda
\end{split}
\end{align*} 
with 
\begin{align*}
\widetilde{\psi}(\tau,z,x)
&
=\frac{(\sqrt{\tau})^{d-2}\psi_+(\sqrt{\tau}x)
-(\sqrt z)^{d-2}\psi_+(\sqrt z x)}{\tau-z}.
\end{align*}
Then the function $\eta_\gamma(z,x)$ is obviously analytic in $(z,x)\in\Delta(\gamma^2)\times\mathbb C^d$,
and we can write 
\begin{align*}
\begin{split}
k_\gamma(z,x)
&
=
\tfrac14(\sqrt z)^{d-2}\psi_+(\sqrt zx)
\Bigl(\int_0^{\gamma^2}\frac{1}{\tau-z}\,\mathrm d\tau
-\int_{-\gamma^2}^0\frac{1}{\tau-z}\,\mathrm d\tau\Bigr)
+\eta_\gamma(z,x)
\\&
=
-\tfrac12(\sqrt z)^{d-2}\psi_+(\sqrt zx)\log z
+\chi_\gamma(z,x)
,
\end{split}
\end{align*}
where 
$$\chi_\gamma(z,x)
=\tfrac14(\sqrt z)^{d-2}\psi_+(\sqrt zx)
\bigl[\log (z-\gamma^2)
+\log (z+\gamma^2)\bigr]
+\eta_\gamma(z,x).$$
Hence we are done.
\end{proof}

\subsubsection{Odd-odd signature}

Finally we prove Theorem~\ref{150803236} for $(p,q)$ odd-odd.
In this case the functions $g_{\pm,\gamma}(\tau,x)$ from \eqref{1605082}
are dependent on choice of contours of integrations.
Here let us always choose contours inside $\mathbb C\setminus (-\infty,0]$.

\begin{lemma}\label{1605222022b}
Let $(p,q)$ be odd-odd, and $\gamma>0$.
Then there exists an entire function $h_{\gamma}(w,x)$ 
in $(w,x)\in \mathbb C\times\mathbb C^d$ 
such that the functions $g_{\pm,\gamma}(\tau,x)$ defined by 
\eqref{1605082} for $(\tau,x)\in(\mathbb C\setminus (-\infty,0])\times\mathbb C^d$ 
have the expressions
\begin{align}
g_{\pm,\gamma}(\tau,x)
&=
\tau^{d-2}\phi_\pm(\tau x)\log(\gamma/\tau)
+ h_{\gamma}(\pm\tau^2,x).
\label{16051223b}
\end{align}
\end{lemma}
\begin{proof}
Let us change variables in \eqref{1605082}, and decompose
\begin{align}
\begin{split}
g_{\pm,\gamma}(\tau,x)
&=
\tau^{d-2}\phi_\pm(\tau x)\log(\gamma/\tau)
\\&\phantom{{}={}}
+\tau^{d-2}
\int_\tau^\gamma\frac{f_\pm(\sigma/\tau,\tau x)-\phi_{\pm}(\tau x)}{\sigma}
\,\mathrm d\sigma
.
\end{split}
\label{170214}
\end{align}
We note that the last integrals of \eqref{170214}
are independent of choice of contours, since the residues at 
$\sigma=0$ are subtracted from the integrands.
Moreover, with the factor $\tau^{d-2}$
apparent singularities at $\tau=0$ in the last terms of 
\eqref{170214}
are in fact removable by noting the expressions \eqref{11.12.12.22.3}.
Hence the last terms of 
\eqref{170214} are entire in $(\tau,x)\in \mathbb C\times \mathbb C^d$.
Furthermore, by \eqref{1702131630}, \eqref{1702131631} and Proposition~\ref{160522193} 
they are even in $\tau\in \mathbb C$.
Thus we may write 
\begin{align}
h_{\gamma}(\tau^2,x)
&=
\tau^{d-2}
\int_\tau^\gamma\frac{f_+(\sigma/\tau,\tau x)-\phi_+(\tau x)}{\sigma}
\,\mathrm d\sigma
,
\label{16051423}
\end{align}
where $h_\gamma(w,x)$ is 
an entire function in $(w,x)\in \mathbb C\times\mathbb C^d$.
Now by \eqref{1702131630}, \eqref{1702131631}, Proposition~\ref{160522193}
and \eqref{11.12.11.13.40}
we can verify
\begin{align*}
h_{\gamma}(-\tau^2,x)
&=
(\mathrm i\tau)^{d-2}\int_{\mathrm i\tau}^{\gamma}
\frac{f_+\bigl(\sigma/(\mathrm i\tau),\mathrm i\tau x\bigr)-\phi_+(\mathrm i\tau x)}{\sigma}
\,\mathrm d\sigma
\\&=
\tau^{d-2}\int_{\mathrm i\tau}^{\gamma}
\frac{f_-(\sigma/\tau,\tau x)-\phi_-(\tau x)}{\sigma}\,\mathrm d\sigma
\\&=
\tau^{d-2}\int_{\mathrm i}^{\gamma/\tau}
\frac{f_-(\sigma,\tau x)-\phi_-(\tau x)}{\sigma}\,\mathrm d\sigma
\\&=
\tau^{d-2}\int_1^{\gamma/\tau}
\frac{f_-(\sigma,\tau x)-\phi_-(\tau x)}{\sigma}\,\mathrm d\sigma
\\&=
\tau^{d-2}
\int_\tau^\gamma\frac{f_-(\sigma/\tau,\tau x)-\phi_-(\tau x)}{\sigma}
\,\mathrm d\sigma
.
\end{align*}
This verifies the assertion.
\end{proof}

\begin{proof}[Proof of Theorem~\ref{150803236} for $(p,q)$ odd-odd]
Substitute \eqref{16051223b} into \eqref{11.12.12.15.4},
and change variables of the integrations. 
Using also Proposition~\ref{160522193}, we have
\begin{align*}
k_\gamma(z,x)
&
=
\int_0^\gamma
\frac{\tau^{d-1}\phi_+(\tau x)\log(\gamma/\tau)}{\tau^2-z}\,\mathrm d\tau
+\int_0^\gamma\frac{\tau h_{\gamma}(\tau^2,x)}{\tau^2-z}\,\mathrm d\tau
\\&\phantom{{}={}}
-\int_0^\gamma
\frac{\tau^{d-1}\phi_-(\tau x)\log(\gamma/\tau)}{\tau^2+z}\,\mathrm d\tau
-\int_0^\gamma\frac{\tau h_{\gamma}(-\tau^2,x)}{\tau^2+z}\,\mathrm d\tau.
\\&
=
\tfrac14\int_0^{\gamma^2}
\frac{(\sqrt{\lambda})^{d-2}\phi_+(\sqrt{\lambda}x)\log(\gamma^2/\lambda)}{\lambda-z}\,\mathrm d\lambda
\\&\phantom{{}={}}
+\tfrac14\int_{-\gamma^2}^0
\frac{(\sqrt{\lambda})^{d-2}\phi_+(\sqrt{\lambda}x)\log(-\gamma^2/\lambda)}{\lambda-z}\,\mathrm d\lambda
+\tfrac12\int_{-\gamma^2}^{\gamma^2}\frac{h_{\gamma}(\lambda,x)}{\lambda-z}\,\mathrm d\lambda.
\end{align*}
Set 
\begin{align*}
\begin{split}
\eta_\gamma(z,x)
&=
\tfrac14\int_0^{\gamma^2}
\widetilde{\phi}(\lambda,z,x)\log\bigl({\gamma^2}/{\lambda}\bigr)\,\mathrm d\lambda
\\&\phantom{{}={}}
{}
+\tfrac14\int_{-\gamma^2}^0
\widetilde{\phi}(\lambda,z,x)\log\bigl(-{\gamma^2}/{\lambda}\bigr)\,\mathrm d\lambda
+\tfrac12\int_{-\gamma^2}^{\gamma^2}\frac{h_{\gamma}(\lambda,x)}{\lambda-z}\,\mathrm d\lambda
\end{split}
\end{align*} 
with 
\begin{align*}
\widetilde{\phi}(\lambda,z,x)
&
=\frac{(\sqrt{\lambda})^{d-2}\phi_+(\sqrt{\lambda}x)-(\sqrt z)^{d-2}\phi_+(\sqrt zx)}{\lambda-z}.
\end{align*}
Obviously, $\eta_\lambda(z,x)$ is analytic in $(z,x)\in\Delta(\gamma^2)\times\mathbb C^d$.
Then we can write 
\begin{align*}
\begin{split}
k_\gamma(z,x)
&
=
\tfrac14(\sqrt z)^{d-2}\phi_+(\sqrt zx)
\Bigl(
\int_0^{\gamma^2}\frac{\log(\gamma^2/\lambda)}{\lambda-z}\,\mathrm d\lambda
+\int_{-\gamma^2}^0\frac{\log(-\gamma^2/\lambda)}{\lambda-z}\,\mathrm d\lambda\Bigr)
\\&\phantom{{}={}}
+\eta_\gamma(z,x)
.
\end{split}
\end{align*}
Finally let us compute the last two integrals.
We first compute for  $w\in\mathbb C_+\setminus \overline{\Delta(1)}$, and then
\begin{align*}
\int_0^1\frac{\log\lambda}{\lambda-w}\,\mathrm d\lambda
-\int_0^1\frac{\log\lambda}{\lambda+w}\,\mathrm d\lambda
&=2w\int_0^1\frac{\log\lambda}{\lambda^2-w^2}\,\mathrm d\lambda
\\&=\frac1{2w}\int_0^\infty
\frac{\sigma\mathrm e^{-\sigma/2}}{1-\mathrm e^{-\sigma}/w^2}\,\mathrm d\sigma
\\&=
\sum_{k=0}^\infty
\frac1{2w^{2k+1}}\int_0^\infty\sigma\mathrm e^{-(k+1/2)\sigma}\,\mathrm d\sigma
\\&=
\sum_{k=0}^\infty\frac2{(2k+1)^2w^{2k+1}}
\\&=
\mathop{\mathrm{Li}}\nolimits_2(1/w)
-\mathop{\mathrm{Li}}\nolimits_2(-1/w),
\end{align*}
where $\mathop{\mathrm{Li}}\nolimits_2w$ is the dilogarithm defined for 
$|w|<1$ by 
$$\mathop{\mathrm{Li}}\nolimits_{2}(w)
=\sum_{k=1}^\infty \frac{w^k}{k^2},$$
and analytically continued for $w\in \mathbb{C}\setminus [1,\infty)$.
Then by analytic continuation we have for $w\in \mathbb C_+$
\begin{align}
\int_0^1\frac{\log\lambda}{\lambda-w}\,\mathrm d\lambda
-\int_0^1\frac{\log\lambda}{\lambda+w}\,\mathrm d\lambda
=\mathop{\mathrm{Li}}\nolimits_2(1/w)
-\mathop{\mathrm{Li}}\nolimits_2(-1/w).
\label{1702149}
\end{align}
If we use the identity
\begin{equation*}
\mathop{\mathrm{Li}}\nolimits_2(1/w)
=-\mathop{\mathrm{Li}}\nolimits_2(w)-\tfrac12(\log(-w))^2-\tfrac{\pi^2}{6}
\quad \text{for }\mathop{\mathrm{Im}}w>0,
\end{equation*}
see \cite[(3.2)]{LCM}, then we can rewrite \eqref{1702149} as 
\begin{equation*}
\int_0^1\frac{\log\lambda}{\lambda-w}\,\mathrm d\lambda
-\int_0^1\frac{\log\lambda}{\lambda+w}\,\mathrm d\lambda
=\mathrm i\pi\log w+\tfrac{\pi^2}{2}
-\mathop{\mathrm{Li}}\nolimits_2(w)
+\mathop{\mathrm{Li}}\nolimits_2(-w)
.
\end{equation*}
By change of variables we obtain 
\begin{align*}
&\int_0^{\gamma^2}\frac{\log(\gamma^2/\lambda)}{\lambda-z}\,\mathrm d\lambda
+\int_{-\gamma^2}^0\frac{\log(-\gamma^2/\lambda)}{\lambda-z}\,\mathrm d\lambda
\\&
=
-\mathrm i\pi\log({z}/{\gamma^2})
-\tfrac{\pi^2}{2}
+\mathop{\mathrm{Li}}\nolimits_2(z/\gamma^2)
-\mathop{\mathrm{Li}}\nolimits_2(-z/\gamma^2)
,
\end{align*}
and hence 
\begin{align*}
k_\gamma(z,x)
&=
-\tfrac{\mathrm i\pi}4(\sqrt z)^{d-2}\phi_+(\sqrt z x)\log({z}/{\gamma^2})
+\chi_\gamma(z,x)
\end{align*}
with
\begin{align*}
\chi_\gamma(z,x)=
\tfrac14(\sqrt z)^{d-2}\phi_+(\sqrt z x)\bigl[
-\tfrac{\pi^2}{2}
+\mathop{\mathrm{Li}}\nolimits_2(z/\gamma^2)
-\mathop{\mathrm{Li}}\nolimits_2(-z/\gamma^2)
\bigr]
+\eta_\gamma(z,x)
\end{align*}
Hence we are done.
\end{proof}

\section{Proofs for the discrete Laplacian}\label{17021822}

The proofs are almost the same as Section~\ref{17021720}.
We separate the elliptic and hyperbolic thresholds.

\subsection{Thresholds at end points}

We first consider the elliptic threshold. It suffices to consider $(p,q)=(d,0)$.
Note that in this case the decomposition \eqref{16060111}
has only two terms $k_0(z,n)$ and $k_1(z,n)$.
In the integral \eqref{160531} we introduce the spherical coordinates
\begin{align*}
\xi=\rho\omega,\quad (\rho,\omega)\in [0,2]\times S^{d-1},
\end{align*}
and then we have 
\begin{align}
k_1(z,n)
&{}=\int_0^2 
\frac{\rho^{d-1}e(\rho,n)}{\rho^2-z}\,\mathrm d\rho;\label{11.12.12.14.52bb}
\end{align}
with
\begin{align}
e(\rho,n)&{}=(2\pi)^{-d}\int_{S^{d-1}}
\prod_{j=1}^d\frac{\exp \bigl(2\mathrm in_j\arcsin(\rho\omega_j/2)\bigr)}{
(1-\rho^2\omega_j^2/4)^{1/2}}
\,\mathrm dS(\omega).
\label{11.12.7.22.52bb}
\end{align}
We recall that the principal branch is being chosen for the above $(1-\rho^2\omega_j^2/4)^{1/2}$.

\begin{proposition}\label{11.9.10.10.7}
The function $e(\rho,n)$ satisfies the identities 
\begin{align}
e(\rho,n)=2E^{(1)}(\rho^2,n),
\quad
-\triangle e(\rho,n)=\rho^2e(\rho,n),
\label{1702172040}
\end{align}
where $\triangle$ denotes the discrete Laplacian with respect to $n$.
\end{proposition}
\begin{proof}
In this proof the branch of square root is the principal one
with cut along the negative real axis.
The latter identity of \eqref{1702172040} follows by directly computations
employing the expression \eqref{11.12.7.22.52bb}.
Hence, also by \eqref{170218}, it suffices to show that 
\begin{align}
e(\rho,n)
=
\frac{2}{2^d\pi^{d/2}}
\sum_{\alpha\in\mathbb Z_+^d}
\frac{\prod_{j=1}^{d}(1/2-n_j)_{\alpha_j}(1/2+n_j)_{\alpha_j}}{\Gamma(|\alpha|+d/2)}
\frac{\rho^{2|\alpha|}}{4^{|\alpha|}\alpha!},
\label{11.9.14.11.1b}
\end{align}
where $(\nu)_k:=\Gamma(\nu+k)/\Gamma(\nu)$ denotes the Pochhammer symbol, 
see~\cite[Definition~5.2(iii)]{Olver:2010:NHMF}.
The expression \eqref{11.12.7.22.52bb} is obviously analytic in $\rho\in\Delta(2)$.
Moreover, the odd parts of the integrand in $\omega_j$
do not contribute to the integral \eqref{11.12.7.22.52bb}, and we can write
\begin{align}
e(\rho,n)
&{}=(2\pi)^{-d}
\int_{S^{d-1}}
\prod_{j=1}^d
\frac{\cos \bigl(2n_j\arcsin (\rho\omega_j/2)\bigr)}{(1-\rho^2\omega_j^2/4)^{1/2}}
\,\mathrm dS(\omega).
\label{160601}
\end{align}
In order to expand \eqref{160601} in $\rho\in\Delta(2)$, 
let us utilize well-known formulas concerning the Chebyshev polynomials 
$T_n$ and the hypergeometric function $F$:
\begin{align*}
\cos (n\theta)=T_n(\cos \theta),\quad T_n(x)=F(n,-n;1/2;(1-x)/2).
\end{align*}
See \cite[Chapter 16]{Olver:2010:NHMF} for the notation and the results used.
Then 
\begin{align}
\cos \bigl(2n_j\arcsin(\rho\omega_j/2)\bigr)
&=T_{n_j}\bigl(\cos\bigl(2\arcsin(\rho\omega_j/2)\bigr)\bigr)\notag
\\&
=F(n_j,-n_j;{1}/{2};\rho^2\omega_j^2/4).
\label{11.9.14.10.43}
\end{align}
We also make use of the Euler transformation formula for the hypergeometric function:
\begin{align*}
F(a,b;c;z)
=(1-z)^{c-a-b}F(c-a,c-b;c;z),
\end{align*}
so that the factors of integrand of \eqref{160601} may be rewritten as 
\begin{align*}
\frac{\cos \bigl(2n_j\arcsin (\rho\omega_j/2)\bigr)}{(1-\rho^2\omega_j^2/4)^{1/2}}
&=
F(1/2-n_j,1/2+n_j;1/2;\rho^2\omega_j^2/4)
\\&
=
\sum_{k=0}^\infty
\frac{(1/2-n_j)_k(1/2+n_j)_k}{(1/2)_k}
\frac{\rho^{2k}\omega_j^{2k}}{4^kk!}
\\&
=
\sum_{k=0}^\infty\frac{(1/2-n_j)_k(1/2+n_j)_k}{(2k)!}\rho^{2k}\omega_j^{2k}
.
\end{align*}
Then we can proceed, using also \eqref{1606011}, 
\begin{align*}
e(\rho,n)
&
=(2\pi)^{-d} 
\sum_{\alpha\in\mathbb Z_+^d}
\rho^{2|\alpha|}
\frac{\prod_{j=1}^{d}(1/2-n_j)_{\alpha_j}(1/2+n_j)_{\alpha_j}}{(2\alpha)!}
\int_{S^{d-1}}\omega^{2\alpha}\,\mathrm dS(\omega)
\\&
=\frac{2}{2^d\pi^{d/2}}
\sum_{\alpha\in\mathbb Z_+^d}
\rho^{2|\alpha|}
\frac{\prod_{j=1}^{d}(1/2-n_j)_{\alpha_j}(1/2+n_j)_{\alpha_j}}{
4^{|\alpha|}\alpha!\Gamma(|\alpha|+d/2)}
.
\end{align*}
Hence we obtain \eqref{11.9.14.11.1b}. 
\end{proof}

\begin{proof}[Proof of Theorem~\ref{1606011129} for thresholds at end points]
The expression \eqref{11.12.12.14.52bb} is very similar to \eqref{1508011300},
so that 
we can repeat the proof of Theorem~\ref{150803236} for the elliptic case.
We omit the detail.
\end{proof}

\subsection{Embedded thresholds}

Next, we consider the embedded thresholds: $p,q\ge 1$.
In the integral \eqref{160531} let us introduce the split spherical coordinates
$$
\xi'=\rho'\omega',\ \xi''=\rho''\omega'';\quad
(\rho',\rho'',\omega',\omega'')\in [0,2]^2\times S^{p-1}\times S^{q-1},
$$
and then we have
\begin{align}
k_1(z,n)
&{}=\int_{\genfrac{}{}{0pt}{}{\rho',\rho''\ge 0,}{\rho'+\rho''<2}} 
\frac{\rho'^{p-1}\rho''^{q-1}e'(\rho',n')e''(\rho'',n'')
}{\rho'^2-\rho''^2-(z-4q)}\,\mathrm d\rho'\mathrm  d\rho'',\label{11.12.12.14.52b}
\end{align}
where, the principal branch being chosen for the square root,
\begin{align*}
e'(\rho',n')&{}=(2\pi)^{-p}\int_{S^{p-1}}
\prod_{j=1}^p\frac{\exp \bigl(2\mathrm in'_j\arcsin(\rho'\omega'_j/2)\bigr)}{
(1-\rho'^2\omega_j'^2/4)^{1/2}}
\,\mathrm dS(\omega'),
\\
e''(\rho'',n'')&{}=(2\pi)^{-q}\int_{S^{q-1}}
\prod_{j=1}^q\frac{\exp \bigl(2\mathrm in''_j\arccos(\rho''\omega''_j/2)\bigr)}{
(1-\rho''^2\omega_j''^2/4)^{1/2}}
\,\mathrm dS(\omega'')
\\&{}
=(-1)^{|n''|}(2\pi)^{-q}\int_{S^{q-1}}
\prod_{j=1}^q\frac{\exp \bigl(2\mathrm in''_j\arcsin(\rho''\omega''_j/2)\bigr)}{
(1-\rho''^2\omega_j''^2/4)^{1/2}}
\,\mathrm dS(\omega'').
\end{align*}
It is clear that the identities similar to those in Proposition~\ref{11.9.10.10.7} 
hold also for a function of the form above. 
As in Section~\ref{160514},
we split the integration region of \eqref{11.12.12.14.52b}:
\begin{align}
\begin{split}
k_1(z,n)
&{}=\int_{\genfrac{}{}{0pt}{}{\rho'\ge \rho''\ge 0,}{\rho'+\rho''<2}} 
\frac{\rho'^{p-1}\rho''^{q-1}e'(\rho',n')e''(\rho'',n'')}{\rho'^2-\rho''^2-(z-4q)}
\,\mathrm d\rho'\mathrm d\rho''
\\&\phantom{={}}
+\int_{\genfrac{}{}{0pt}{}{\rho''>\rho'\ge 0,}{\rho'+\rho''<2}} 
\frac{\rho'^{p-1}\rho''^{q-1}e'(\rho',n')e''(\rho'',n'')}{\rho'^2-\rho''^2-(z-4q)}
\,\mathrm d\rho'\mathrm d\rho'',
\end{split}
\label{160501b}
\end{align}
and change the variables by 
\begin{align*}
\rho'=\tfrac12{\tau(\sigma\pm\sigma^{-1})},
\quad
\rho''=\tfrac12{\tau(\sigma\mp\sigma^{-1})},
\end{align*}
respectively.
If we introduce the functions
\begin{align}
\begin{split}
f_{\pm}(\sigma,\tau,n)&=
\frac{(\sigma\pm\sigma^{-1})^{p-1}(\sigma\mp\sigma^{-1})^{q-1}}{2^{d-2}}
\\&\phantom{{}={}}
\cdot
e'\bigl(\tfrac12{\tau(\sigma\pm\sigma^{-1})},n'\bigr)
e''\bigl(\tfrac12{\tau(\sigma\mp\sigma^{-1})},n''\bigr)
\end{split}
\label{11.12.12.22.3c}
\end{align}
and 
\begin{align}
g_{\pm}(\tau,n)&=\tau^{d-2}\int_1^{2/\tau}
\frac{ f_\pm(\sigma,\tau,n)}{\sigma}\,\mathrm d\sigma
,
\label{1605082b}
\end{align}
then \eqref{160501b} can be rewritten as
\begin{align}
\begin{split}
k_1(z,n)
&
=\int_0^2\frac{\tau g_{+}(\tau,n)}{\tau^2-(z-4q)}\,\mathrm d\tau
-\int_0^2\frac{\tau g_{-}(\tau,n)}{\tau^2+(z-4q)}\,\mathrm d\tau.
\end{split}\label{11.12.12.15.4c}
\end{align}

Let us set 
\begin{equation*}
D=\bigl\{(\sigma,\tau)\in (\mathbb C\setminus\{0\})\times \mathbb C;\ \tau(\sigma+\sigma^{-1})\in \Delta(4),\ 
\tau(\sigma-\sigma^{-1})\in \Delta(4)\bigr\}.
\end{equation*} 
Obviously, the functions $f_{\pm}(\sigma,\tau,n)$ from \eqref{11.12.12.22.3c}
are continued as single-valued analytic function in $(\sigma,\tau)\in D$.
However, the analytic continuations of the functions $g_{\pm}(\tau,n)$ from \eqref{1605082b}
are not necessarily single-valued.
Let us set for $(\tau,n)\in \Delta(2)\times \mathbb Z^d$
\begin{align}
\phi_\pm(\tau,n)
&
=\frac1{2\pi\mathrm i}\int_{|\sigma|=1}
\frac{f_\pm(\sigma,\tau,n)}{\sigma}\,\mathrm d\sigma
,
\label{11.12.11.13.40c}
\\
\psi_\pm(\tau,n)
&=\int_{\Gamma(1)}\frac{f_\pm(\sigma,\tau,n)-\phi_\pm(\tau,n)}{\sigma}\,\mathrm d\sigma
\label{16051123c}
\end{align}
with $\Gamma(1)=\{\mathrm e^{\mathrm i\theta}\in\mathbb C;\ \theta\in[0,\pi/2]\}$.

\begin{proposition}\label{16052215b} 
The functions $\phi_{\pm}(\tau,n)$ and $\psi_{\pm}(\tau,n)$ 
are analytic in $\tau\in \Delta(2)$, 
and satisfy
$$
\phi_+(\tau,n)=\frac{4}{\mathrm i\pi}E^{(1)}(\tau^2,n)
,\quad 
\phi_-(\tau,n)=\frac{4\mathrm i^{p-q}}{\mathrm i\pi}E^{(1)}(-\tau^2,n)
,\quad 
\psi_\pm(\tau,n)=0$$
if $(p,q)$ is odd-odd, and 
$$\phi_\pm(\tau,n)=0,\quad
\psi_+(\tau,n)=2E^{(1)}(\tau^2,n)
,\quad
\psi_+(\tau,n)=2\mathrm i^{p-q}E^{(1)}(-\tau^2,n)
$$
otherwise.
In addition, they also satisfy
\begin{align}
-\triangle\phi_\pm(\tau,n)&=(\pm \tau^2+4q)\phi_\pm(\tau,n),
\label{1605161925c}
\\
-\triangle\psi_\pm(\tau,n)&=(\pm \tau^2+4q)\psi_\pm(\tau,n),
\label{1605161925b}
\end{align}
where $\triangle$ denotes the discrete Laplacian with respect to $n$.
\end{proposition}
\begin{proof}
The identities \eqref{1605161925c} and \eqref{1605161925b} are verified with ease,
if we note the expressions \eqref{11.12.11.13.40c}, \eqref{16051123c} and the identities
\begin{align*}
(-\triangle\mp \tau^2-4q)f_\pm(\sigma,\tau,n)=0.
\end{align*}
The rest of the assertions are verified completely the same manner as in
the proof of Proposition~\ref{160522193}. We omit the details.
\end{proof}

\begin{proof}[Proof of Theorem~\ref{1606011129} for embedded thresholds]
If we argue similarly to the proof of Theorem~\ref{150803236} for hyperbolic thresholds, 
we can deduce the expressions
\begin{align*}
k_1(w+4q,n)=\mathrm i\pi(\sqrt w)^{d-2}E^{(1)}(w,n)+\chi^{(1)}(w,n)
\end{align*}
if $d$ is odd, and 
\begin{align*}
k_1(w+4q,n)
&=
-(\sqrt w)^{d-2}(\log w)E^{(1)}(w,n)+\chi^{(1)}(w,n)
\end{align*}
if $d$ is even.
The construction of this section so far is concerned with the component $k_1(z,n)$
in \eqref{16060111},
but we can in fact do similar constructions for $k_l(z,n)$, $l=2,\dots,L$.
Then by summing up we obtain the assertion. We omit the details.
\end{proof}

\bigskip
\noindent
\subsubsection*{Acknowledgements} 
KI would like to thank Yoshiaki Goto for pointing out an error in the draft on the Lauricella hypergeometric functions. 
AJ would like to thank Vojkan Jak\v{s}i\'{c} for making him aware of the paper~\cite{Poulin} and Jan Derezi\'{n}ski for comments on hypergeometric functions. 
KI was supported by JSPS KAKENHI Grant Numbers JP25800073 and 17K05325.
The authors were partially supported by the Danish Council for Independent Research $|$ Natural Sciences, Grant DFF--4181-00042.

\end{document}